\documentclass[10pt,a4paper]{article}

\usepackage[utf8]{inputenc}

\usepackage{graphicx}
\usepackage{hyperref}

\usepackage{multicol}
\usepackage{multirow}

\usepackage{url}

\usepackage{tikz}
\usetikzlibrary{arrows.meta}

\usepackage{subcaption}

\usepackage{amsmath,amssymb,amsfonts}
\usepackage{algorithm}
\usepackage{algpseudocode}

\usepackage{mathtools,amsmath,amsfonts,amsthm,paralist}
\newtheorem{theorem}{Theorem}

\usepackage{listings}
\usepackage{color}
\usepackage{xcolor}
\newcommand{\flo}[1]{\textcolor{black}{#1}}
\newcommand{\ac}[1]{\textcolor{black}{#1}}

\begin{document}

\title{Algorithm-Based Fault Tolerance\\for Parallel Stencil Computations}
\author{Aurélien Cavelan\\
\textit{University of Basel, Switzerland}\\
\textit{aurelien.cavelan@unibas.ch}
\and
Florina M. Ciorba\\
\textit{University of Basel, Switzerland}\\
\textit{florina.ciorba@unibas.ch}
}

\maketitle

\begin{abstract}
The increase in HPC systems size and complexity, together with increasing on-chip transistor density, power limitations, and number of components, render modern HPC systems subject to soft errors. 
Silent data corruptions (SDCs) are typically caused by such soft errors in the form of bit-flips in the memory subsystem and hinder the correctness of scientific applications.
This work addresses the problem of protecting a class of iterative computational kernels, called stencils, against SDCs when executing on parallel HPC systems.
Existing SDC detection and correction methods are in general either inaccurate, inefficient, or targeting specific application classes that do not include stencils.
This work proposes a novel algorithm-based fault tolerance (ABFT) method to protect scientific applications that contain arbitrary stencil computations against SDCs.
The ABFT method can be applied both \emph{online} and \emph{offline} to accurately detect and correct SDCs in 2D and 3D parallel stencil computations.
We present a formal model for the proposed method including theorems and proofs for the computation of the associated checksums as well as error detection and correction.
We experimentally evaluate the use of the proposed ABFT method on a real 3D stencil-based application (HotSpot3D) via a \mbox{fault-injection}, detection, and correction campaign.
Results show that the proposed ABFT method achieves less than $8\%$ overhead compared to the performance of the unprotected stencil application. 
Moreover, it accurately detects and corrects SDCs. 
While the offline ABFT version corrects errors more accurately, it may incur a small additional overhead than its online counterpart.
\end{abstract}


\section{Introduction}
\label{sec.introduction}

High Performance Computing (HPC) systems have rapidly grown in size and complexity over the last years.
With increasing on-chip transistor density, power limitations, and number of components, it is expected that the number of soft errors continues to grow in modern systems~\cite{ExascaleSoftwareStudy,MeasuringAndUnderstandingExtremeScale,LessonsLearned}.
Such errors can trigger bit-flips in the memory subsystem, leading to silent data corruptions (SDCs) which represent a major threat to data correctness~\cite{Snir_FailuresExascale,IESP-toward,JSFI14}.
Several phenomena cause soft errors, such as cosmic rays, aging components or packaging pollution, among others~\cite{Ogorman94,Ziegler98}.
Modern systems are equipped with error correcting codes (ECC) to protect the memory subsystems against data corruptions.
Recent studies suggest that such mechanisms may not prevent data corruptions from occurring at extreme scales~\cite{LaercioNeutron,BGomezLargeScaleStudy,GPUerrors,CosmicRaysDontStrikeTwice} and in particular in DRAM devices~\cite{MemoryErrorsGoodBadUgly}.

In this work, we focus on a class of iterative kernels that update array elements using a fixed computational pattern that repeats over the entire domain, called a \emph{stencil}.
Stencil-based computational kernels arise from the use of discretization methods such as finite differences, finite volumes, or finite elements. 
As such, they are commonly used to solve partial differential equations (PDEs) in the fields of computational fluid dynamics (CFD), cosmology, and combustion~\cite{OptimizationAndPerformanceModelingOfStencil,AnIntroductionToNumericalMethodsAndAnalaysis}. 
Other common examples of stencil-based kernels include the Jacobi kernel, the Gauss–Seidel method, and image processing. 
The simple structure of stencil-based kernels can be exploited to produce highly efficient computations. 
However, their high memory bandwidth requirements~\cite{ImplicitAndExplicitOptimizationsForStencil} make them particularly vulnerable to SDCs~\cite{BGomezLargeScaleStudy,MemoryErrorsGoodBadUgly}.

Triple modular redundancy (TMR)~\cite{Lyons1962}, together with other \mbox{\emph{full-replication}} of entire applications~\cite{Fiala12Detection} and \emph{selective-replication} of parts of an application~\cite{Berrocal17Toward} are the most general and non-intrusive approaches to detect and correct SDCs in scientific applications.
However, such methods are prohibitively expensive in terms of additional required computing resources and time.
Therefore, many application-specific detectors have been proposed as alternatives to redundancy to lower the cost of error detection. 
Specifically, \emph{data analytics-}based error detection techniques detect outliers by relying on application-specific properties, such as spatial and/or temporal data smoothness. 
\emph{Interpolation-}based error detectors employ techniques such as time series prediction and spatial multivariate interpolation~\cite{PPoPP14, Gomez15Detecting, Gomez15Exploiting} to interpolate the next value of a data-point; they offer broad error detection coverage at low cost, but typically have a lower precision compared to full-replication approaches.
A further review of relevant error detection methods is presented in Section~\ref{sec.related}.

Algorithm-based fault tolerance (ABFT) is another widely used fault tolerance approach that exploits algorithmic features to encode a small amount of redundancy into the computation, typically in the form of invariant checksums. 
These redundant features can subsequently be used to detect and correct one or multiple errors. 
ABFT has primarily been developed for linear-algebra kernels \cite{Kuang1984,Shantharam2012, ABFTDenseMatrixFactorization,NewSumOnlineABFT} due to its very low overhead, high detection rate, and the possibility to correct errors on-the-fly.
However, ABFT methods are only available for very specific problems which do not include stencil computations.

In this work, we introduce a novel ABFT method for parallel stencil computations on regular 2D and 3D computational grids.
Adapting ABFT to stencil computations is not straightforward.
The main challenge is to preserve the invariance of the checksum vectors (one for the domain rows and one for the domain columns).
Variations in the checksums make it impossible to directly compare the checksum vectors between two consecutive stencil iterations as it is the case with most ABFT approaches.
\textbf{In this work, we propose a novel method to interpolate the checksum vectors of the stencil domain based solely on the checksum vectors from the previous stencil iteration} (as illustrated later in Figure~\ref{fig.abft}).
Our approach only requires one checksum vector to \emph{detect} errors, and two checksum vectors to \emph{detect and correct} errors online after each stencil iteration.
However, the second checksum is only needed if an error has been detected in the first place, which means only one checksum must be computed every iteration.
To minimize the overhead, we show that the computation of the first checksum can be made by adding just one extra operation to stencil kernel.
In addition, we propose an extension of the method to detect errors \emph{offline}, after the application has completed or for a given detection period.
The offline method cannot correct errors.
Therefore, it must be coupled with an existing error correction mechanism. 
In this work we employ the standard checkpointing and recovery approach to correct errors offline.

Even though the proposed ABFT scheme can be applied on the entire 2D or 3D computational domain, consistently computing the checksums across threads or processes on a parallel and distributed system can lead to unwanted and costly synchronization and communication. 
\textbf{To avoid such additional costs, we propose an implementation that independently performs all the steps of the proposed approach: checksum computation, interpolation, detection, and correction, within each thread or process}.
This makes the scheme \emph{intrinsically parallel}, so that it can be applied both within a shared memory and a distributed memory system.

We implement our method on a well-known parallel stencil mini-app and perform experiments to assess both the overhead of the method and the accuracy of the correction mechanism.
Early results show that in an error-free environment, the online and the offline implementations of the method achieve similar performance, with less than $8\%$ overhead.
In case of silent errors during the execution, the overhead of the offline method increases significantly due to the high recovery cost, while the overhead of the online method remains almost unchanged. 
While the online method leads to acceptable result accuracy, the offline method is more reliable and better at correcting outliers.

To the best of our knowledge, this is the \emph{first ABFT method for arbitrary stencil-based kernel computations with arbitrary boundary conditions}.
This work brings forward the following contributions:
\begin{itemize}
\item An online detection and correction ABFT method for 2D and 3D stencil computations based on a \textbf{novel interpolation method} 
\item An offline extension of the ABFT method to detect errors after the application complete or with a given period,
\item A formal model, theorems and proofs for the interpolation of the checksums and the detection of errors, and 
\item Experiments on a real 3D stencil-based application using fault-injection.
\end{itemize}

This work is organized as follows:
Section~\ref{sec.related} reviews the related work, then Section~\ref{sec.abft2d} introduces the proposed online detection and correction ABFT method for 2D stencils. Section~\ref{sec.offline} describes the offline variant of the ABFT method. 
Section~\ref{sec.experiments} presents the experimental results, followed by the conclusion and ideas for future work outlined in Section~\ref{sec.conclusion}.

\section{Related Work}
\label{sec.related}

Several approaches have been developed to detect and correct SDCs in scientific applications.
Triple modular redundancy (TMR)~\cite{Lyons1962} is the most general and non-intrusive approach to detect and correct soft errors in scientific applications.
In general, detectors that either employ \emph{full replication} of entire applications~\cite{Fiala12Detection} or \emph{selective replication} of parts of an application~\cite{Berrocal17Toward} offer the highest detection rate with the lowest number of false-positives.
However, such approaches are prohibitively expensive in terms of additional required computing resources and time.


ABFT has primarily been developed for linear-algebra kernels \cite{Kuang1984,bosilca2009algorithm,Shantharam2012,onlineABFT,ABFTDenseMatrixFactorization,NewSumOnlineABFT}.
Another ABFT scheme has recently been proposed to correct soft errors online in the Fast Fourier Transform~\cite{CorrectingFFT}.
An Algorithm-Based Error Detection (ABED) method for multigrid solvers is studied in~\cite{ABFTMultigrid}. 

An error correction mechanism for stencils has previously been proposed~\cite{ABFTStencil}, wherein authors exploits stencil locality to reduce the re-execution cost after an error has been detected.
However, no error detector was provided.

Silent error detectors based on \emph{data analytics} use several interpolation techniques, such as time series prediction~\cite{Berrocal15Lightweight} 
and multivariate interpolation~\cite{PPoPP14,Gomez15Detecting,Gomez15Exploiting}, to \emph{interpolate} the next value of a computational point based on spatial and temporal data smoothness.
In particular, a method based on multivariate interpolation has been proposed to predict, detect, and correct SDCs in stencil computations~\cite{Gomez15Detecting}.
While their approach is subject to false-positives under chaotic phenomena (e.g. shockwaves), the authors show that their method can detect and correct errors with a magnitude above $10^{-2}$ for a single data point, without mentioning the associated overhead in terms of execution time.
In contrast, \emph{our method accurately detects and corrects errors with a magnitude above $10^{-5}$, independently of the simulated phenomenon}.
Furthermore, our method does not raise any false-positives and achieves an overhead of less than $8\%$ compared to the unprotected application.


\section{Online ABFT for Stencils}
\label{sec.abft2d}


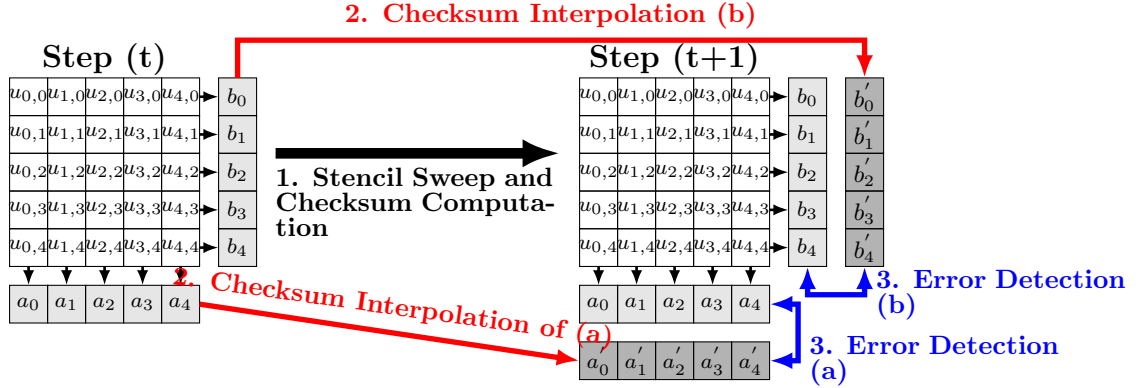
\begin{figure*}[!htb]
\begin{center}
\begin{tikzpicture}[scale=0.5]
  \foreach \x in {0,...,4}
    \foreach \y in {0,...,4} 
       {\pgfmathtruncatemacro{\ly}{4-\y}
       \node [] at (\x+0.5,\y+0.5) {\footnotesize$u_{\x,{\ly}}$};}

  \foreach \x in {0,...,4}
    \foreach \y [count=\yi] in {0,...,4}  
      \draw (\x,\y) rectangle (\x+1,\y+1);

  \foreach \x in {0,...,4}
       \node [] at (\x+0.5,-1) {\footnotesize$a_{\x}$};
  \foreach \x in {0,...,4}
      \draw[fill=black, fill opacity=0.1] (\x,-1.5) rectangle (\x+1,-0.5);

  \foreach \y in {0,...,4}
  	   {\pgfmathtruncatemacro{\ly}{4-\y}
       \node [] at (6,\y+0.5) {\footnotesize$b_{\ly}$};}
  \foreach \y in {0,...,4}
      \draw[fill=black, fill opacity=0.1] (5.5,\y) rectangle (6.5,\y+1);

  \foreach \x in {0,...,4}
      \draw [black,line width=1pt,{}-{Latex[length=2mm]}] (\x+0.5,0) -- (\x+0.5,-0.5);
  \foreach \y in {0,...,4}
      \draw [black,line width=1pt,{}-{Latex[length=2mm]}] (5,\y+0.5) -- (5.5,\y+0.5);

\def\gx{15}
	
	\draw (2.5,5.5) node[font=\footnotesize\bfseries]{\large Step (t)};
	\draw (\gx+2.5,5.5) node[font=\footnotesize\bfseries]{\large Step (t+1)};

  \foreach \x in {0,...,4}
    \foreach \y in {0,...,4} 
       {\pgfmathtruncatemacro{\ly}{4-\y}
       \node [] at (\gx+\x+0.5,\y+0.5) {\footnotesize$u_{\x,{\ly}}$};}

  \foreach \x in {0,...,4}
    \foreach \y [count=\yi] in {0,...,4}  
      \draw (\gx+\x,\y) rectangle (\gx+\x+1,\y+1);

  \foreach \x in {0,...,4}
       \node [] at (\gx+\x+0.5,-1) {\footnotesize$a_{\x}$};
  \foreach \x in {0,...,4}
      \draw[fill=black, fill opacity=0.1] (\gx+\x,-1.5) rectangle (\gx+\x+1,-0.5);

  \foreach \x in {0,...,4}
       \node [] at (\gx+\x+0.5,-1.5+-1) {\footnotesize$a^{'}_{\x}$};
  \foreach \x in {0,...,4}
      \draw[fill=black, fill opacity=0.3]  (\gx+\x,-1.5+-1.5) rectangle (\gx+\x+1,-1.5+-0.5);

  \foreach \y in {0,...,4}
  	   {\pgfmathtruncatemacro{\ly}{4-\y}
       \node [] at (\gx+6,\y+0.5) {\footnotesize$b_{\ly}$};}
  \foreach \y in {0,...,4}
      \draw[fill=black, fill opacity=0.1] (\gx+5.5,\y) rectangle (\gx+6.5,\y+1);

  \foreach \y in {0,...,4}
  	   {\pgfmathtruncatemacro{\ly}{4-\y}
       \node [] at (\gx+1.5+6,\y+0.5) {\small$b^{'}_{\ly}$};}
  \foreach \y in {0,...,4}
      \draw[fill=black, fill opacity=0.3] (\gx+1.5+5.5,\y) rectangle (\gx+1.5+6.5,\y+1);

    \foreach \x in {0,...,4}
      \draw [black,line width=1pt,{}-{Latex[length=2mm]}] (\gx+\x+0.5,0) -- (\gx+\x+0.5,-0.5);
    \foreach \y in {0,...,4}
      \draw [black,line width=1pt,{}-{Latex[length=2mm]}] (\gx+5,\y+0.5) -- (\gx+5.5,\y+0.5);

   \draw [red,line width=2pt,{}-{Latex[length=3mm]}] (5,-1.0) -- (15,-2.5)
   	node[midway,above,text=red,rotate=-8,font=\footnotesize\bfseries]{\normalsize 2. Checksum Interpolation of (a)};

   \draw [red,line width=2pt,{}-{Latex[length=3mm]}] (6,5.0) -- (6,6) -- (\gx+7.5,6) 
   	node[midway,above,text=red,rotate=0,font=\footnotesize\bfseries]{\normalsize 2. Checksum Interpolation (b)}
   	-- (\gx+7.5,5.0);

   \draw [blue,line width=2pt,{Latex[length=3mm]}-{Latex[length=3mm]}] (\gx+5,-1.0) -- (\gx+5.75,-1.0) -- (\gx+5.75,-2.5) 
   	node[right,text=blue,rotate=0,font=\footnotesize\bfseries,text width=3.3cm]{\normalsize 3. Error Detection (a)}
   	-- (\gx+5,-2.5);

   	\draw [blue,line width=2pt,{Latex[length=3mm]}-{Latex[length=3mm]}] (\gx+6,0) -- (\gx+6,-0.75) -- (\gx+7.5,-0.75) 
   	node[below,right,text=blue,rotate=0,font=\footnotesize\bfseries,text width=3.3cm]{\normalsize 3. Error Detection (b)}
   	-- (\gx+7.5,0);

   	\draw [black,line width=4pt,{}-{Latex[length=5mm]}] (7,3) -- (\gx-0.5,3)
   	node[midway,below,text=black,rotate=0,font=\footnotesize\bfseries,text width=3.75cm]{\normalsize 1. Stencil Sweep and Checksum Computation};


\end{tikzpicture}
\end{center}
\caption{Overview of the proposed ABFT approach on a 2D 5 $\times$ 5 stencil domain with both row and column checksum vectors $a$ and $b$. Assuming the initial checksum vectors at step t are correct, the approach has three steps: (1) perform the stencil sweep and compute the new checksums at step t+1, (2) interpolate the checksum vectors at step t+1 from those at step t, and (3) compare the checksum vectors to detect errors. The interpolation method is introduced in Section~\ref{sec.interpolation}.}
\label{fig.abft} 
\end{figure*}

This section introduces the online version of the proposed ABFT scheme for arbitrary stencil-based kernel computations on multi-dimensional grids.
The presentation of the proposed ABFT method in this section concentrates on 2D stencils.
Without loss of generality, the method is applicable to 3D stencils by simply applying the 2D scheme on every layer of the 3D domain.
The ABFT scheme can be applied online and offline, both to 2D and 3D stencils.
For 3D stencils, the complexity of the method increases linearly with the number of layers.

The proposed method, illustrated in Figure~\ref{fig.abft}, follows the main steps of most ABFT schemes: computing checksums, checking checksums, and using checksums for correcting errors when needed.

We introduce a set of \emph{general notation} to model arbitrary stencil-based kernel computations in Section~\ref{sec.sweep} and show \emph{how to compute the checksum vectors} in Section~\ref{sec.checksum}.
Due to the fact that such checksums are not invariant, it is not possible to directly compare their values between two stencil iterations. 
We solve this problem by introducing a novel method to interpolate the new checksum vectors from those of the previous iteration.
Having both the original and the interpolated checksum vectors, we can directly compare their values, and raise an error if they differ.
We present the \emph{novel checksum interpolation method} in Section~\ref{sec.interpolation} and the \emph{error detection method} in Section~\ref{sec.detection}.
Finally, if an error is detected, both checksum vectors (row and column) must contain one erroneous value, the position of this erroneous value giving out the exact location of the error in the computational domain. 
From there, it is possible to retrieve the correct value by subtracting the erroneous value from both checksum vectors and by solving the resulting equation. 
This \emph{error correction mechanism} is presented in Section~\ref{sec.correction}.

\subsection{Stencil Sweep}
\label{sec.sweep}

We introduce a set of general notation to describe an arbitrary \emph{stencil sweep operation} on a 2D domain or sub-domain (chunk or block) of size $n_x \times n_y$. 
Let $u_{x,y}^{(t)}$ denote the state of the stencil function $u$ at coordinate $(x,y)$ and at iteration $t$, and let $\{i,j,w\} \in S$ denote the set of stencil points with relative coordinates $(i,j)$ and weight coefficient $w$. 
For example, a 2D \mbox{4-point} stencil that computes the average of its four neighbors would be defined as:
$S = \{(0,-1,0.25), (-1,0,0.25), (1,0,0.25), (0,1,0.25)\}$. 
A stencil sweep operation consists of updating each stencil (grid) point by applying the stencil operator as follows:
\begin{align}
\label{eq.update}
u_{x,y}^{(t+1)} &= C_{x,y} + \sum_{\{i,j,w\}\in S} w \cdot u_{x+i,y+j}^{(t)} \ ,
\end{align}
where for each point, we perform a weighted summation of the neighboring points and add an optional constant term $C_{x,y}$ to account for arbitrary stencil computations (e.g., localized heat source or sink). Note that the weight coefficient $w$ is individual to every stencil point.

\subsection{Checksum Computation}
\label{sec.checksum}

The first step to achieving ABFT is to compute checksums in two separate checksum vectors (one for the domain rows and one for the domain columns).
Let $a$ and $b$ denote the row and the column checksum vectors, respectively. 
We write:
\begin{align}
a_{x}^{(t+1)} &= \sum_{y=0}^{n_y} u_{x,y}^{(t+1)} \label{eq.checksum_x} \\
b_{y}^{(t+1)} &= \sum_{x=0}^{n_x} u_{x,y}^{(t+1)} \label{eq.checksum_y}\ ,
\end{align}
where $x$ and $y$ indicate the row index and the column index in the row and column checksum vectors, respectively, and $t$ denotes the stencil iteration number.

\medskip 
\noindent\textbf{Implementation.} 
Stencil kernels typically consist of a few optimized operations repeated numerous times.
Adding even a single operation can have a significant impact on the performance of the stencil-based scientific application. 
In Section~\ref{sec.detection}, we show that a single checksum vector is enough for SDC detection, and that the second checksum vector only needs to be computed in case of error, to enable correction. 
We recommend that only one checksum vector is computed during the stencil sweep. 
This can efficiently be realized by adding a single addition operation to the kernel operation and ensuring that the operation minimizes the number of cache misses, as shown in Figure~\ref{fig.sweep}.

\medskip 
\noindent\textbf{Limitations.}
In a parallel environment, each thread must be assigned a different checksum vector. Fortunately, the online approach only requires one checksum vector to be computed, allowing up $n_x$ threads to work in parallel (as opposed to $n_x \times n_y$). The same is true for 3D stencil computations, where up to $n_z \times n_x$ threds can work in parallel. Note that this does not represent a major limitation for sufficiently large domains.

\begin{figure}[!htb]
\begin{center}
\begin{lstlisting}
<@\texttt{\textcolor{blue}{// Parallel stencil sweep and checksum computation}}@> 
#pragma omp parallel for
for(int y = 0; y < ny; y++) {
  for(int x = 0; x < nx; x++) {
    int c, w, e, n, s;
    c =  x + y * nx + z * nx * ny; // <@\texttt{center}@>
    w = (x == 0)    ? c : c - 1;   // <@\texttt{west}@>
    e = (x == nx-1) ? c : c + 1;   // <@\texttt{east}@>
    n = (y == 0)    ? c : c - nx;  // <@\texttt{north}@>
    s = (y == ny-1) ? c : c + nx;  // <@\texttt{south}@>
    u[t+1][c] = C[c] + w1 * u[t][c] + w2 * u[t][w] 
        + w3 * u[t][e] + w4 * u[t][s] + w5 * u[t][n];
    <@\texttt{\textcolor{blue}{// Column checksum computation}} @> 
    <@\texttt{\textcolor{blue}{b[t+1][y] += u[t+1][c];}} @> 
  }
}
\end{lstlisting}
\end{center}
\caption{Example implementation of a 2D five-point stencil sweep operation (lines 11-12), augmented by the computation of the column checksum vector \texttt{b} (line 14) for the proposed ABFT method. Note that lines 7-10 represent the boundary conditions that define the stencil behavior at the border of the computational domain. Terms \texttt{w1}, \texttt{w2}, $\ldots$, \texttt{w5} represent individual stencil weight \flo{coefficients} and can take arbitrary values.}
\label{fig.sweep}
\end{figure}


\subsection{Checksum Interpolation}
\label{sec.interpolation}

We show next how to interpolate the checksum vectors at stencil iteration $t+1$ from the checksum vectors at iteration $t$. 
The interpolation is made by \emph{applying the original 2D stencil kernel to the 1D checksum vectors}. 
The following theorem formally describes this process. 
Note the similarity with Equation~\eqref{eq.update}.

\begin{theorem}
\label{th.checksum_rec}
The checksum vectors at stencil iteration $t+1$ of a 2D stencil domain $u$ that is swept according to Equation~\eqref{eq.update} can be computed from the checksum vectors at iteration $t$ as follows:
\begin{align}
&a_{x}^{(t+1)} = c_x + \sum_{\{i,j,w\}\in S} w \left(a_{x+i}^{(t)} + \alpha_{x+i,j}^{(t)}\right) \label{eq.checksum_rec_x},\\
&b_{y}^{(t+1)} = c_y + \sum_{\{i,j,w\}\in S} w \left(b_{y+j}^{(t)} + \beta_{i,y+j}^{(t)}\right) \label{eq.checksum_rec_y}
\end{align}
with $c_x = \sum_{y=0}^{n_y} C_{x,y}$ and the boundary conditions:
\begin{align*}
\alpha_{x,j} &=\left\{
\begin{array}{ll}
  \sum_{y=j}^{-1} u_{x,y}^{(t)} - \sum_{y={n_y+j}+1}^{n_y} u_{x,y}^{(t)} & \text{if } j < 0, \\
  \sum_{y=n_y+1}^{n_y+j} u_{x,y}^{(t)} - \sum_{y=0}^{j-1} u_{x,y}^{(t)} & \text{if } j > 0,
\end{array}\right.\\
\beta_{i,y} &=\left\{
\begin{array}{ll}
  \sum_{x=i}^{-1} u_{x,y}^{(t)} - \sum_{x={n_x+i+1}}^{n_x} u_{x,y}^{(t)} & \text{if } i < 0, \\
  \sum_{x=n_x+1}^{n_x+i} u_{x,y}^{(t)} - \sum_{x=0}^{i-1} u_{x,y}^{(t)} & \text{if } i > 0.
\end{array}\right.
\end{align*}
where $S$ is the set of stencil points ${(i,j)}$ each characterized by its weight coefficient $w$ and by its relative position $i$ and $j$ to the $x$ and $y$ coordinates, respectively.
Computing the checksum vectors takes at most $O(k n_y)$ and $O(k n_x)$ computational time and requires at most $O(n_x)$ and $O(n_y)$ additional memory space, respectively, where $k = |S|$ denotes the number of stencil points.
\end{theorem}

\begin{proof}
The goal of this proof is to show that we can recursively express the row checksum vector $a_{x}^{(t+1)}$ as a function of $a_{x}^{(t)}$. 
The idea is to show that we can apply the 2D stencil kernel to the checksum vectors from iteration $t$ to obtain the checksum vectors at iteration $t+1$. 
We start by expanding Equation~\eqref{eq.checksum_x}. 
Using the definition of the stencil sweep operation in Equation~\eqref{eq.update}, we write:
\begin{align*}
a_{x}^{(t+1)} &= \sum_{y=0}^{n_y} \left(C_{x,y} + \sum_{\{i,j,w\}\in S} w \cdot u_{x+i,y+j}^{(t)}\right).
\end{align*}
Then, by rearranging the terms in the above equation, we obtain:
\begin{align*}
a_{x}^{(t+1)} &= \sum_{y=0}^{n_y} C_{x,y} + \sum_{\{i,j,w\}\in S} \left(\sum_{y=0}^{n_y} w \cdot u_{x+i,y+j}^{(t)}\right) \\
&=c_x + \sum_{\{i,j,w\}\in S} \left(w \cdot \sum_{y=0}^{n_y} u_{x+i,y+j}^{(t)}\right) \ ,
\end{align*}
with $c_x = \sum_{y=0}^{n_y} C_{x,y}$. 
Then, rewriting the second summation, we obtain:
\begin{align}
\label{eq.checksum_rearranged}
&= c_x + \sum_{\{i,j,w\}\in S} \left(w \cdot \sum_{y=j}^{n_y+j} u_{x+i,y}^{(t)}\right) \ .
\end{align}
From here, the goal is to rewrite $\sum_{y=j}^{n_y+j} u_{x+i,y}^{(t)}$ to retrieve $\sum_{y=0}^{n_y} u_{x+i,y}^{(t)}$, a.k.a. $a_{x+i}^{(t)}$. 
We distinguish the following three cases:
\begin{align*}
&\sum_{y=j}^{n_y+j} u_{x+i,y}^{(t)} = \\
&\qquad\left\{
\begin{array}{ll}
  \begin{array}{r}
  \sum_{y=j}^{-1} u_{x,y}^{(t)} + a_{x+i}^{(t)} - \sum_{y={n_y+j+1}}^{n_y} u_{x,y}^{(t)} 
  \end{array} & \text{if } j < 0, \\
  \begin{array}{r}
  \sum_{y=0}^{n_y} u_{x,y}^{(t)} + a_{x+i}^{(t)} - \sum_{y=0}^{j-1} u_{x,y}^{(t)} 
  \end{array} & \text{if } j > 0, \\
  \begin{array}{r}
  a_{x+i}^{(t)}
  \end{array} & \text{otherwise.} \\
\end{array}\right.
\end{align*}
Rearranging the terms above, we obtain:
\begin{align}
&\sum_{y=j}^{n_y+j} u_{x+i,y}^{(t)} = a_{x+i}^{(t)} + \alpha_{x+i,j}
\label{eq.temp_uxi}
\end{align}
with:
\begin{align*}
&\alpha_{x,j} =\left\{
\begin{array}{ll}
  \sum_{y=j}^{-1} u_{x+i,y}^{(t)} - \sum_{y={n_y+j}+1}^{n_y} u_{x+i,y}^{(t)} & \text{if } j < 0, \\
  \sum_{y=n_y+1}^{n_y+j} u_{x+i,y}^{(t)} - \sum_{y=0}^{j-1} u_{x+i,y}^{(t)} & \text{if } j > 0 .
\end{array}\right.
\end{align*}
Finally, substituting $\sum_{y=j}^{n_y+j} u_{x+i,y}^{(t)}$ in Equation~\eqref{eq.checksum_rearranged} by Equation~\eqref{eq.temp_uxi}, we retrieve Equation~\eqref{eq.checksum_rec_x}. The same approach is used to derive Equation~\eqref{eq.checksum_rec_y} for the column checksum vector $b$.

\medskip
\noindent\textbf{Complexity.} The number of operations necessary to compute $\alpha_{x+i,j}$ depends on the size of the stencil. 
It requires at most $O(k)$ computational time, where $k = |S|$ denotes the number of stencil points. 
It takes $O(n_y)$ time to compute $c_x$, which is constant and can be pre-computed. 
Overall, computing the row checksum vector $a$ requires at most $O(n_y k^2)$ computational time and $O(n_x)$ additional memory space.
The same complexity analysis can be made for computing and storing the column checksum vector $b$, which concludes the proof.
\end{proof}

\medskip
\noindent\textbf{Dealing with Boundary Conditions.}
Equation~\eqref{eq.checksum_rec_x} and Equation~\eqref{eq.checksum_rec_y} describe the interpolation process to compute the row and column checksum vectors, respectively. 
Note that compared to the original stencil operator (see Equation~\eqref{eq.update}), these equations have the additional term $\alpha$ for the row checksum and the additional term $\beta$ for the column checksum. 
These terms arise from the boundary conditions.

Indeed, due to the fixed computational pattern of the stencil, almost all stencil points arrive to have the same contribution to the final checksums. 
Almost every point in the domain is added the same fixed number of times to its corresponding row and column checksums.
This is true for all points in the domain, with the exception of the points that are subject to boundary conditions, where a slightly different stencil kernel is used (see for example the listing in Figure~\ref{fig.sweep}).

If the stencil is configured to use periodic or bounce-back boundary conditions, all points may end up having the same contribution factor to the row and vector checksums. 
In that case the terms $\alpha$ and $\beta$ are both zero, leading to the following simplified version of Equation~\eqref{eq.checksum_rec_x} and Equation~\eqref{eq.checksum_rec_y}, much like the original stencil operator:
\begin{align}
&a_{x}^{(t+1)} = c_x + \sum_{\{i,j,w\}\in S} w \cdot a_{x+i}^{(t)} \label{eq.checksum_rec_x_simple} \\
&b_{y}^{(t+1)} = c_y + \sum_{\{i,j,w\}\in S} w \cdot b_{y+j}^{(t)} \label{eq.checksum_rec_y_simple} \ .
\end{align}

Alternatively, if constant boundary values are used, then the first terms in $\alpha$ and $\beta$ become constant and can be pre-computed. 
If empty boundaries are used (e.g., setting all points beyond the boundary of the computational domain to $0$), then the first terms in $\alpha$ and $\beta$ can simply be discarded.


\medskip 
\noindent\textbf{Implementation.} The listing in Figure~\ref{fig.interpolation} shows a possible implementation of the interpolation process for the column checksum vector $b$. 
The implementation is straightforward and follows Equation~\eqref{eq.checksum_rec_y_simple}. 
The original stencil kernel is slightly modified to work with the 1D checksum vectors. 
In this example, the kernel relies on a bounce-back boundary condition, effectively making $\alpha$ and $\beta$ equal to zero. Note that boundary conditions still need to be accounted for in the modified kernel.

\begin{figure}
\begin{center}
\begin{lstlisting}
<@\texttt{\textcolor{blue}{// Interpolation of the column checksum}}@>
for (int x=0; x<nx; x++) {
    int c, w, e, n, s;
    c = x;                         // <@\texttt{\textcolor{black}{center}}@> 
    w = (x == 0)    ? c : c - 1;   // <@\texttt{\textcolor{black}{west}}@> 
    e = (x == nx-1) ? c : c + 1;   // <@\texttt{east}@> 
    n = c;                         // <@\texttt{north}@> 
    s = c;                         // <@\texttt{south}@> 
    <@\texttt{\textcolor{blue}{bcheck[c] = c[x] + w1 * b[t][c] + w2 * b[t][c]}}@>
    <@\texttt{\textcolor{blue}{        + w3 * b[t][c] + w4 * b[t][s] + w5 * b[t][n];}}@>
}
<@\texttt{\textcolor{red}{//b[t+1]: computed as in previous listing}}@>
<@\texttt{\textcolor{red}{compare(b[t+1], bcheck); // Next: compare and detect}}@>
\end{lstlisting}
\end{center}
\caption{Implementation of the checksum interpolation (see Equation~\eqref{eq.checksum_rec_y_simple}). A slightly modified version of the 2D kernel from Figure~\ref{fig.sweep} is applied on the 1D checksum vector \texttt{b} from the previous iteration to obtain the new checksum vector \texttt{bcheck}, which can later be directly compared to the checksum vector \texttt{b} from the current iteration \texttt{(t+1)}.}
\label{fig.interpolation} 
\end{figure}

\subsection{Error Detection}
\label{sec.detection}

In this section, we present the error detection method. The following theorem is used to compare the new checksum vectors computed via Equations~\eqref{eq.checksum_x} and~\eqref{eq.checksum_y} against the interpolated checksum vectors computed via Equations~\eqref{eq.checksum_rec_x} and~\eqref{eq.checksum_rec_y}.

\begin{theorem}
Assuming all stencil points are correct at stencil iteration $t$, then for any stencil sweep operation based on Equation~\eqref{eq.update}, the interpolated checksums computed at iteration $t+1$ with Equations~\eqref{eq.checksum_rec_x} and~\eqref{eq.checksum_rec_y} are equal to the checksums computed at iteration $t+1$ directly from the output data according to Equations~\eqref{eq.checksum_x} and~\eqref{eq.checksum_y}, i.e., the following equalities hold:
\begin{align*}
a_{x}^{(t+1)} &= \sum_{y=0}^{n_y} u_{x,y}^{(t+1)} &= \sum_{\{i,j,w\}\in S} w \left(a_{x+i}^{(t)} + \alpha_{x+i,j}^{(t)}\right),\\
b_{y}^{(t+1)} &= \sum_{x=0}^{n_x} u_{x,y}^{(t+1)} &= \sum_{\{i,j,w\}\in S} w \left(b_{y+j}^{(t)} + \beta_{i,y+j}^{(t)}\right),
\end{align*}
if and only if there are no SDCs during the computation nor SDCs that cancel each other out.
\label{th.detection}
\end{theorem}

\begin{proof}
We assume that the initial data at stencil iteration $t=0$ are correct and that the initial checksum is also correct.
Then, assuming that the stencil and the checksums at $t=i$ are also correct, we need to study the possible values of 
(1)~the checksum vectors computed from the updated stencil at $t=i+1$ using Equation~\eqref{eq.checksum_x} or Equation~\eqref{eq.checksum_y}), and 
(2)~the interpolated checksum vectors, computed from the previous (correct) checksum vectors at $t=i$ using Equations~\eqref{eq.checksum_rec_x} and~\eqref{eq.checksum_rec_y}:
\begin{itemize}
\item An error that occurs in the domain at $t=i$ (after the checksum at $t=i$ has been computed) will cause (1)~to be incorrect and (2)~to be correct.
\item An error that occurs in the domain at $t=i+1$ will cause (1)~to be incorrect, and (2)~to be correct.
\item An error that occurs in (1)~will cause (1)~to be incorrect and (2)~to be correct.
\item An error that occurs in (2)~will cause (1)~to be correct and (2)~to be incorrect.
\end{itemize}
Note that if several errors occur, both (1)~and (2)~may be incorrect but not equal, therefore, an error will still be detected. 
With very low probability, errors (e.g. bit-flips) might strike different points in of the domain in such a way that errors cancel each other out in the final checksums, rendering all errors undetectable. This concludes the proof.
\end{proof}

\medskip 
\noindent\textbf{Implementation.} Due to round-off or discretization errors when using the IEEE Standard 754 TM floating-point arithmetic\footnote{\url{http://ieeexplore.ieee.org/document/4610935/}}, checksums as defined in Equation~\eqref{eq.checksum_x}, \eqref{eq.checksum_y}, \eqref{eq.checksum_rec_x} and \eqref{eq.checksum_rec_y} may not be equal, even when no soft errors \flo{occurred} during execution\footnote{\ac{Phenomena such as subtractive cancellations, absorptions, or underflows can lead to accuracy losses.}}.

In this work, we compare the relative error between each element of the checksum vectors, $a_{x}^{(t+1)}$ and $b_{y}^{(t+1)}$, computed according to Equations~\eqref{eq.checksum_x} and~\eqref{eq.checksum_y}, respectively, against the checksum vectors $a_{x}^{(t+1)'}$ and $b_{y}^{(t+1)'}$ computed according to Equations~\eqref{eq.checksum_rec_x} and~\eqref{eq.checksum_rec_y}.
The relative error is computed as follows: 
\begin{align*}
\left|\frac{a_{x}^{(t+1)'}}{a_{x}^{(t+1)}} - 1\right| \ \text{and } \left|\frac{b_{y}^{(t+1)'}}{b_{y}^{(t+1)}} - 1\right| \ ,
\end{align*}
and an error flag is set if the relative error is higher than some \emph{detection threshold} $\epsilon$, which depends on the domain, chunk, or block size on which the method is applied. 
The approximation error proportionally increases with the domain size.


When an error is detected, its location (row and column) in the domain is given out by the value of $y$ and $x$, respectively, and can be stored in an array to be later used for correction.

As per Theorem~\ref{th.detection}, it is only necessary to perform the detection on one of the two checksums. 
If one or more errors are detected, only then it is necessary to interpolate the other checksum and to perform detection to obtain the location of the errors both in $x$ and $y$. 
The listing in Figure~\ref{fig.detect} shows the implementation corresponding to the error detection process.

\begin{figure}
\begin{center}
\begin{lstlisting}
<@\texttt{\textcolor{blue}{// Compare checksum vectors bcheck and b}}@>
for (int y=0; y<ny; y++) {
    float rel = fabs(bcheck[y]/b[t+1][y] - 1.0);
    <@\texttt{\textcolor{red}{//Detect error}}@>
    if(rel > EPSILON) {
      erry[error_count_y] = y;
      error_count_y++;
    }
}
\end{lstlisting}
\end{center}
\caption{Error detection is implemented by computing the relative error between the interpolated column checksum \texttt{bcheck} (see Figure~\ref{fig.interpolation}) and the newly computed column checksum \texttt{b} (see Figure~\ref{fig.sweep}). If the error is greater than some \texttt{EPSILON} value, an error is detected in row \texttt{y}. The same method is applied to the row checksum \texttt{a} to obtain the precise location of the error.}
\label{fig.detect} 
\end{figure}

\subsection{Error Correction}
\label{sec.correction}

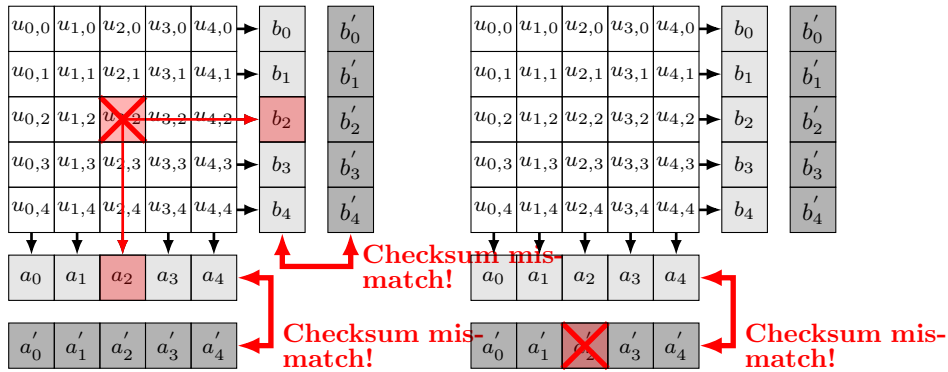
\begin{figure*}[!htb]
\begin{center}
\begin{subfigure}{0.492\linewidth}
\centering
\begin{tikzpicture}[scale=0.6]
\def\gx{0}
  \foreach \x in {0,...,4}
    \foreach \y in {0,...,4} 
       {\pgfmathtruncatemacro{\ly}{4-\y}
       \node [] at (\gx+\x+0.5,\y+0.5) {\footnotesize$u_{\x,{\ly}}$};}

  \foreach \x in {0,...,4}
    \foreach \y [count=\yi] in {0,...,4}  
      \draw (\gx+\x,\y) rectangle (\gx+\x+1,\y+1);

  \foreach \x in {0,...,4}
       \node [] at (\gx+\x+0.5,-1) {\footnotesize$a_{\x}$};
  \foreach \x in {0,...,4}
      \draw[fill=black, fill opacity=0.1] (\gx+\x,-1.5) rectangle (\gx+\x+1,-0.5);

  \foreach \x in {0,...,4}
       \node [] at (\gx+\x+0.5,-1.5+-1) {\footnotesize$a^{'}_{\x}$};
  \foreach \x in {0,...,4}
      \draw[fill=black, fill opacity=0.3] (\gx+\x,-1.5+-1.5) rectangle (\gx+\x+1,-1.5+-0.5);

  \foreach \y in {0,...,4}
  	   {\pgfmathtruncatemacro{\ly}{4-\y}
       \node [] at (\gx+6,\y+0.5) {\footnotesize$b_{\ly}$};}
  \foreach \y in {0,...,4}
      \draw[fill=black, fill opacity=0.1] (\gx+5.5,\y) rectangle (\gx+6.5,\y+1);

  \foreach \y in {0,...,4}
  	   {\pgfmathtruncatemacro{\ly}{4-\y}
       \node [] at (\gx+1.5+6,\y+0.5) {\small$b^{'}_{\ly}$};}
  \foreach \y in {0,...,4}
      \draw[fill=black, fill opacity=0.3] (\gx+1.5+5.5,\y) rectangle (\gx+1.5+6.5,\y+1);

    \foreach \x in {0,...,4}
      \draw [black,line width=1pt,{}-{Latex[length=2mm]}] (\gx+\x+0.5,0) -- (\gx+\x+0.5,-0.5);
    \foreach \y in {0,...,4}
      \draw [black,line width=1pt,{}-{Latex[length=2mm]}] (\gx+5,\y+0.5) -- (\gx+5.5,\y+0.5);

   \draw [red,line width=2pt,{}-{}] (\gx+2,3) -- (\gx+3,2);
   \draw [red,line width=2pt,{}-{}] (\gx+2,2) -- (\gx+3,3);

   \draw [red,line width=1pt,{}-{Latex[length=2mm]}] (\gx+2.5,2.5) -- (\gx+2.5,-0.5);
   \draw [red,line width=1pt,{}-{Latex[length=2mm]}] (\gx+2.5,2.5) -- (\gx+5.5,2.5);

   \draw [fill=red,fill opacity=0.3] (\gx+2,2) rectangle (\gx+3,3);
   \draw [fill=red,fill opacity=0.3] (\gx+2,-0.5) rectangle (\gx+3,-1.5);
   \draw [fill=red,fill opacity=0.3] (3.5+\gx+2,3) rectangle (3.5+\gx+3,2);

   \draw [red,line width=2pt,{Latex[length=3mm]}-{Latex[length=3mm]}] (\gx+5,-1.0) -- (\gx+5.75,-1.0) -- (\gx+5.75,-2.5) 
   	node[right,text=red,rotate=0,font=\footnotesize\bfseries,text width=3.3cm]{\normalsize Checksum mismatch!}
   	-- (\gx+5,-2.5);

   	\draw [red,line width=2pt,{Latex[length=3mm]}-{Latex[length=3mm]}] (\gx+6,0) -- (\gx+6,-0.75) -- (\gx+7.5,-0.75) 
   	node[below,right,text=red,rotate=0,font=\footnotesize\bfseries,text width=3.3cm]{\normalsize Checksum mismatch!}
   	-- (\gx+7.5,0);
\end{tikzpicture}
\caption{Scenario 1: error in the stencil domain at position (2,2).}
\end{subfigure}
\begin{subfigure}{0.492\linewidth}
\centering
\begin{tikzpicture}[scale=0.6]
\def\gx{0}

  \foreach \x in {0,...,4}
    \foreach \y in {0,...,4} 
       {\pgfmathtruncatemacro{\ly}{4-\y}
       \node [] at (\gx+\x+0.5,\y+0.5) {\footnotesize$u_{\x,{\ly}}$};}

  \foreach \x in {0,...,4}
    \foreach \y [count=\yi] in {0,...,4}  
      \draw (\gx+\x,\y) rectangle (\gx+\x+1,\y+1);

  \foreach \x in {0,...,4}
       \node [] at (\gx+\x+0.5,-1) {\footnotesize$a_{\x}$};
  \foreach \x in {0,...,4}
      \draw[fill=black, fill opacity=0.1] (\gx+\x,-1.5) rectangle (\gx+\x+1,-0.5);

  \foreach \x in {0,...,4}
       \node [] at (\gx+\x+0.5,-1.5+-1) {\footnotesize$a^{'}_{\x}$};
  \foreach \x in {0,...,4}
      \draw[fill=black, fill opacity=0.3] (\gx+\x,-1.5+-1.5) rectangle (\gx+\x+1,-1.5+-0.5);

  \foreach \y in {0,...,4}
  	   {\pgfmathtruncatemacro{\ly}{4-\y}
       \node [] at (\gx+6,\y+0.5) {\footnotesize$b_{\ly}$};}
  \foreach \y in {0,...,4}
      \draw[fill=black, fill opacity=0.1] (\gx+5.5,\y) rectangle (\gx+6.5,\y+1);

  \foreach \y in {0,...,4}
  	   {\pgfmathtruncatemacro{\ly}{4-\y}
       \node [] at (\gx+1.5+6,\y+0.5) {\small$b^{'}_{\ly}$};}
  \foreach \y in {0,...,4}
      \draw[fill=black, fill opacity=0.3] (\gx+1.5+5.5,\y) rectangle (\gx+1.5+6.5,\y+1);

    \foreach \x in {0,...,4}
      \draw [black,line width=1pt,{}-{Latex[length=2mm]}] (\gx+\x+0.5,0) -- (\gx+\x+0.5,-0.5);
    \foreach \y in {0,...,4}
      \draw [black,line width=1pt,{}-{Latex[length=2mm]}] (\gx+5,\y+0.5) -- (\gx+5.5,\y+0.5);

   \draw [red,line width=2pt,{}-{}] (\gx+2,-3) -- (\gx+3,-2);
   \draw [red,line width=2pt,{}-{}] (\gx+2,-2) -- (\gx+3,-3);

   \draw [fill=red,fill opacity=0.3] (\gx+2,-2) rectangle (\gx+3,-3);

   \draw [red,line width=2pt,{Latex[length=3mm]}-{Latex[length=3mm]}] (\gx+5,-1.0) -- (\gx+5.75,-1.0) -- (\gx+5.75,-2.5) 
   	node[right,text=red,rotate=0,font=\footnotesize\bfseries,text width=3.3cm]{\normalsize Checksum mismatch!}
   	-- (\gx+5,-2.5);

\end{tikzpicture}
\caption{Scenario 2: error in the interpolated checksum vector b'.}
\end{subfigure}
\end{center}
\caption{Two example scenarios where (a) the error strikes in the stencil domain, and (b) the error strikes a checksum vector. In (a) the error location is deduced from the corrupted row and column checksum vectors and can be corrected. In (b), the error can be in either checksum vector $a$ or $a'$: since $a'$ is \flo{no longer needed for verification due to corruption}, only $a$ needs to be recomputed.}
\label{fig.abft.error} 
\end{figure*}

We present here the error correction method. 
Let $a_x^{(t+1)}$ and $b_x^{(t+1)}$ denote the row and column checksums computed via Equations~\eqref{eq.checksum_x} and~\eqref{eq.checksum_y}, respectively.
Let $a_x^{(t+1)'}$ and $b_x^{(t+1)'}$ denote the interpolated row and column checksum computed via Equations~\ref{eq.checksum_x} and~\ref{eq.checksum_y}.

Let $ex$ and $ey$ denote the coordinate of the corrupted stencil point, obtained from Section~\ref{sec.detection}. 
We first subtract the erroneous value $u_{ex,ey}^{(t+1)}$ from either one of the correct checksums. 
The correct value is then obtained by solving the following equation:
\begin{align}
\label{eq.correction}
correct_{ex,ey}^{(t+1)} &= a_{ex}^{(t+1)'} - (a_{ex}^{(t+1)} - u_{ex,ey}^{(t+1)}) \notag \\
&= b_{ey}^{(t+1)'} - (b_{ey}^{(t+1)} - u_{ex,ey}^{(t+1)}) \ .
\end{align}
Note that checksums also need to be updated with the correct value to maintain the correctness of subsequent stencil iterations.

\medskip 
\noindent\textbf{Implementation.} The implementation of error correction follows Equation~\ref{eq.correction}, as shown in the listing in Figure~\ref{fig.correction}. \ac{The correct value can be computed from either checksums. One can compute the \flo{checksum average} or arbitrarily \flo{use one} of the two.}

\begin{figure}
\begin{center}
\begin{lstlisting}
<@\texttt{\textcolor{blue}{// Identification of error location and correction}}@>
for(int i=0; i<error_count; i++) {
    <@\texttt{// Error location}@>
    const int ex = errx[i]; 
    const int ey = erry[i];
    const int c =  ex + ey * nx;

    <@\texttt{// Error correction}@>
    float vx = acheck[ex] - (a[t+1][ex] - u[t+1][c]);
    float vy = bcheck[ey] - (b[t+1][ey] - u[t+1][c]);
    float corrected = (vx + vy) / 2.0;

    <@\texttt{// Update}@>
    a[t+1][ex] += corrected - u[t+1][c];
    b[t+1][ey] += corrected - u[t+1][c];
    u[t+1][c] = corrected;
}
\end{lstlisting}
\end{center}
\caption{Implementation of error correction by subtracting the erroneous value from both checksum vectors and averaging the correct value, which is then injected back into the domain. Note that the checksums also need to be updated to maintain stencil correctness for the next iterations.}
\label{fig.correction} 
\end{figure}

\section{Offline ABFT for Stencils}
\label{sec.offline}

We extend the online version of the proposed ABFT scheme to also enable its use for offline SDCs detection after the application completes or with a given error detection period during the execution of the application.
Section~\ref{sec.offline.detection} describes the extended error detection method and Section~\ref{sec.offline.correction} discusses different error correction approaches.

\subsection{Offline Error Detection}
\label{sec.offline.detection}

Offline error detection is achieved by extending the interpolation process. 
Section~\ref{sec.interpolation} describes the equations used to interpolate the checksum vectors from the past iteration during the current iteration. 
Here, we show how to extend the checksum interpolation with a given period $\Delta$ using the standard checkpointing and recovery approach.
In the absence of errors, checkpoints can safely be performed every $\Delta$ iterations.
In case an error is detected at iteration $t$, we need to recover from the last correct checkpoint at iteration $t-\Delta$, and to begin recomputing from there.
This can be achieved via an iterative process, where Equations~\eqref{eq.checksum_rec_x} and~\eqref{eq.checksum_rec_y} are used to successively interpolate the checksum vectors from iteration $t$ up to $t+\Delta$.

Note that as a result, the asymptotic complexity for the interpolation process is $O(\Delta k n_y)$ and $O(\Delta k n_y)$ for the row and column checksum vectors, respectively. 
This means that the complexity of the detection algorithm for the offline ABFT method scales linearly with the detection period $\Delta$.
However, since error detection only needs to be performed once every $\Delta$ iterations, the cost per iteration is theoretically the same to that of the proposed online ABFT error detection method.

\medskip
\noindent\textbf{Implementation.} The offline ABFT error detection implementation, as shown in Figure~\ref{algo.offline.inteprolation}, assumes the initial checksum $b$ was computed from the last checksum. 
From there, the kernel from Figure~\ref{fig.interpolation} is successively applied $\Delta$ times until the interpolated checksum from the current iteration has been retrieved.
After this step, the same error detection method, as that illustrated in Figure~\ref{fig.detect}, is applied.

\begin{figure}
\begin{center}
\begin{lstlisting}[mathescape=true]
<@\texttt{\textcolor{blue}{bcheck = b[t+1]; // Save current checksum}}@> 
<@\texttt{\textcolor{blue}{\textbf{for} (t=last\_checkpoint; t<last\_checkpoint+dt; t++)\{}}@>
    for (int x=0; x<nx; x++) {
      int c, w, e, n, s;
      c = x;                       <@\texttt{// center}@> 
      w = (x == 0)    ? c : c - 1; <@\texttt{// west}@> 
      e = (x == nx-1) ? c : c + 1; <@\texttt{// east}@> 
      n = c;                       <@\texttt{// north}@> 
      s = c;                       <@\texttt{// south}@> 
      b[t+1][c] = c[x] + w1 * b[t][c] + w2 * b[t][c] 
        + w3 * b[t][c] + w4 * b[t][s] + w5 * b[t][n];
    }
<@\texttt{\textcolor{blue}{\}}}@>
<@\texttt{\textcolor{red}{compare(b[t+1], bcheck); // Next: compare and detect}}@> 
\end{lstlisting}
\end{center}
\caption{Offline checksum interpolation, applying the interpolation kernel from Figure~\ref{fig.interpolation} (lines 2-11) $\Delta$ times to retrieve the current checksum.}
\label{algo.offline.inteprolation}
\end{figure}

Note that this method is more subject to floating-point approximation errors, which may add up to a significant amount, depending on the value of $\Delta$ and on the domain size. 
However, in-memory checkpointing of the domain is relatively fast, which means the detection period, $\Delta$, can be kept small. 
Another approach to avoiding such approximation errors would be to increase the detection threshold $\epsilon$ as shown in Section~\ref{sec.detection} to avoid false-positives by lowering the detection rate.

\subsection{Offline Error Correction}
\label{sec.offline.correction}

Error correction can be achieved in this work by combining existing error correction techniques with the proposed error detection method. 
Checkpointing with rollback recovery~\cite{CL85,uncoordinated,FTI} is the de-facto general-purpose technique to recover from errors or failures in HPC.
Since \flo{the ABFT approach proposed in this work} can be used as a periodic error detector, one can safely checkpoint the state of the domain after a successful detection and correction. This will allow to recover from the state of the domain at that point in case of another error.
Recall that an error correction mechanism specifically dedicated to stencil computations has already been presented~\cite{ABFTStencil}, wherein authors exploit stencil locality to reduce the re-execution cost after an error has been detected. 
Such a method could be used to further lower the cost of re-computation in case of errors.
In the present work, we conduct experiments using the standard checkpoint and recovery method and leave alternatives error correction approaches for future work.

\section{Experiments}
\label{sec.experiments}

In this section, we perform a set of experiments to: 
(a)~assess the overhead of the proposed ABFT method and 
(b)~evaluate the accuracy of the error correction mechanisms.
For this purpose, we compare the performance and accuracy of the \emph{Online ABFT}, \emph{Offline ABFT}, and \emph{No-ABFT} approaches both in an error-free environment and when a silent error (e.g., a bit-flip) is injected during the computation.

These three approaches have been integrated in the HotSpot3D stencil code from the Rodinia benchmark suite\footnote{\url{http://lava.cs.virginia.edu/Rodinia/download_links.htm}}.
HotSpot3D is a widely used simulation tool to estimate processor temperature based on an architectural floorplan and simulated power measurements. 
The ABFT prototype implementation and the early results presented in this work are freely available upon request.

Section~\ref{sec.experiments.setup} describes the parameters used in the experiments while Section~\ref{sec.overhead} presents the main error detection and correction performance and accuracy results.
Section~\ref{sec.accuracy} provides a more in-depth analysis of the error detection and correction accuracy, showing the limitations of each method. 
Finally, Section~\ref{sec.interval} shows the impact of the error detection period on the performance of the \emph{Offline ABFT} method.

\subsection{Experimental Setup}
\label{sec.experiments.setup}

The HotSpot3D stencil code employs OpenMP for node-level parallelization and has built-in functions to measure both the execution time and the accuracy of the computational results. 
The \emph{arithmetic error} between the obtained results and the reference value is computed with an $\textit{l}^{2}$-norm and serves as the measure of accuracy (the reference value is obtained by using a single thread).
Let $v^{ref}$ and $v^{comp}$ denote the reference value and the computed results, respectively. 
Then the arithmetic error is defined as follows:
\begin{align}
\label{eq.l2norm}
error &= \sqrt{\sum_{i=0}^{n_x n_y n_z} \left(v_i^{ref}-v_i^{comp}\right)^2} \ .
\end{align}

Table~\ref{table.experiments.summary} summarizes the parameters used in the experiments.
Each experiment consists of $1,000$ and $100$ repetitions of the HotSpot3D stencil code with each of the three methods: \emph{Online ABFT}, \emph{Offline ABFT} and \emph{No-ABFT} (unprotected application). 
These ABFT methods are applied to 3D cubes (or tiles) of size $512\times512\times8$ and $64\times64\times8$ parallelized over $8$ OpenMP threads on Intel Xeon E5-2640 v4 nodes, where each thread handles one of the $8$ 2D layers of the 3D computational domain.

\ac{Note that each layer uses its own independent checksums and that the proposed ABFT method is applied independently within each layer. Therefore, performance and accuracy are only affected by the size of the tile.}
Choosing the right detection threshold (see Section~\ref{sec.detection}) is critical to avoid detecting false-positives due to small floating-point approximation errors during the interpolation step. \ac{Since the proposed approach is sensitive to the tile size, we propose using small tile sizes to limit the approximation error.}
For the selected tile sizes, setting the detection threshold to $10^{-5}$ was sufficiently high to ensure that no false-positives were reported, and sufficiently low to detect all errors up to the fifth decimal point. 

To assess the performance and accuracy of the ABFT method in an error-prone environment, we simulate SDCs by injecting a single bit-flip in the memory used by the application during the execution. 
The bit-flip is injected during a random stencil iteration ($0\hdots127$), in random point in the computational domain  ($0\hdots512x512x8-1$), and at a random bit position for the particular stencil point ($0\hdots31$).
To ensure that the injected bit-flip has an immediate and visible impact on the stencil results, the injection is performed during the stencil sweep operation, after the stencil point targeted for data corruption has been updated and before it is stored into the domain.

\begin{table}
\begin{center}
\begin{tabular}{|l | c | c |}
\hline
\multirow{ 2}{*}{\textbf{Parameter}} & \multicolumn{2}{|c|}{\textbf{Tile sizes}} \\
\cline{2-3}
& \textbf{64$\times$64$\times$8} & \textbf{512$\times$512$\times$8} \\
\hline
Stencil iterations & 128 & 256 \\
Experiment repetitions & 1,000 & 100 \\
Error detection threshold & $10^{-5}$ & $10^{-5}$ \\
Offline detection period & 16 iterations & 16 iterations \\
\hline
\end{tabular}
\end{center}
\caption{Overview of the main experimental parameters}
\label{table.experiments.summary}
\end{table}

\subsection{Performance and Accuracy Evaluation}
\label{sec.overhead}

In this section, we evaluate the performance and accuracy of the three approaches: \emph{Online ABFT}, \emph{Offline ABFT}, and \emph{\mbox{No-ABFT}}. \mbox{Figures}~\ref{fig.overhead.float} and~\ref{fig.error.float} show the mean execution time and the mean accuracy for the three methods, respectively, in an error-free and in an error-prone scenario with a single randomly injected bit-flip.

\begin{figure}[!htb]
  \begin{center}
      \begin{subfigure}{\linewidth}
      \centering
        \includegraphics[scale=0.52]{{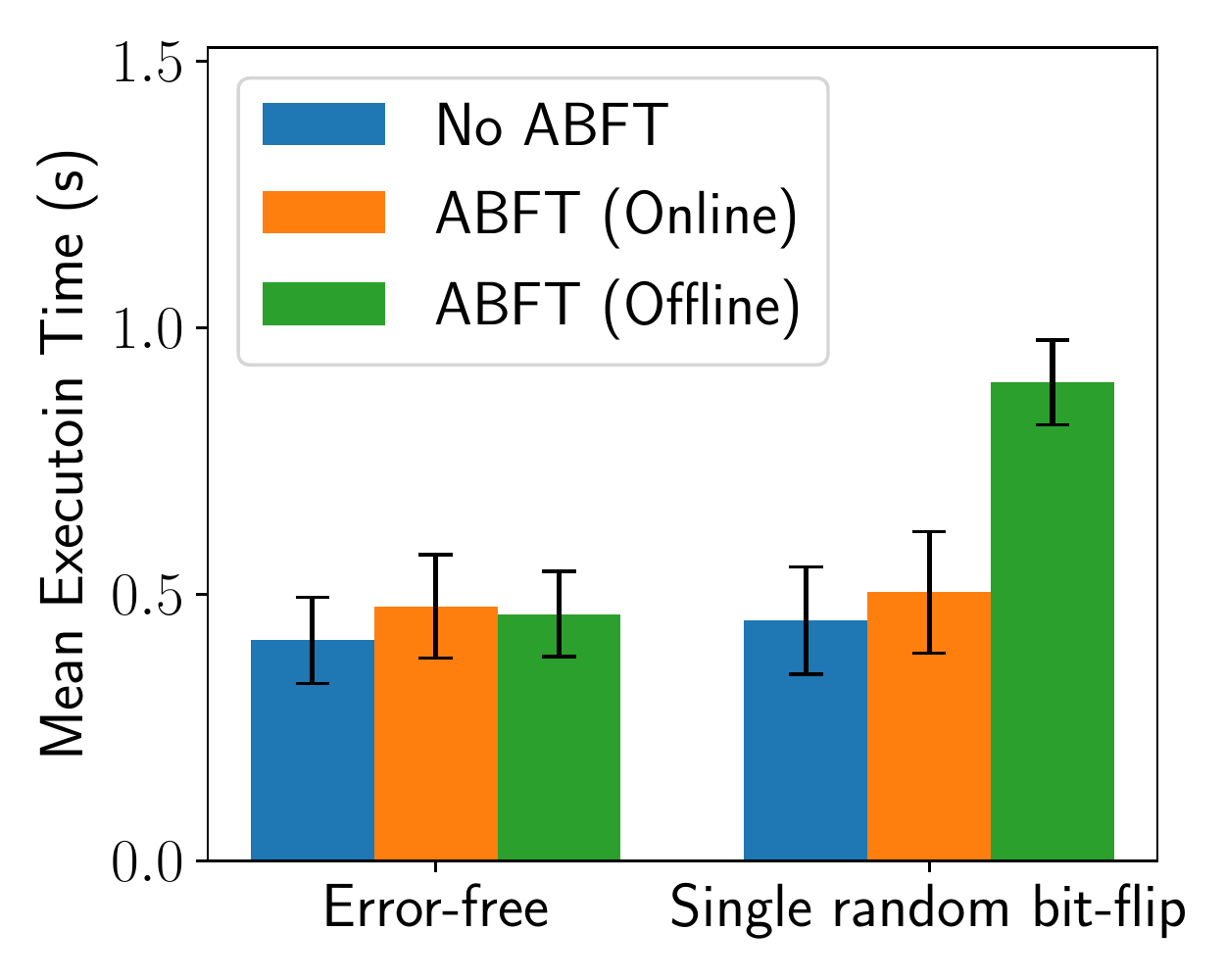}}
        \caption{Tile size \textbf{64$\times$64$\times$8}}
      \end{subfigure}
      \begin{subfigure}{\linewidth}
      \centering
        \includegraphics[scale=0.52]{{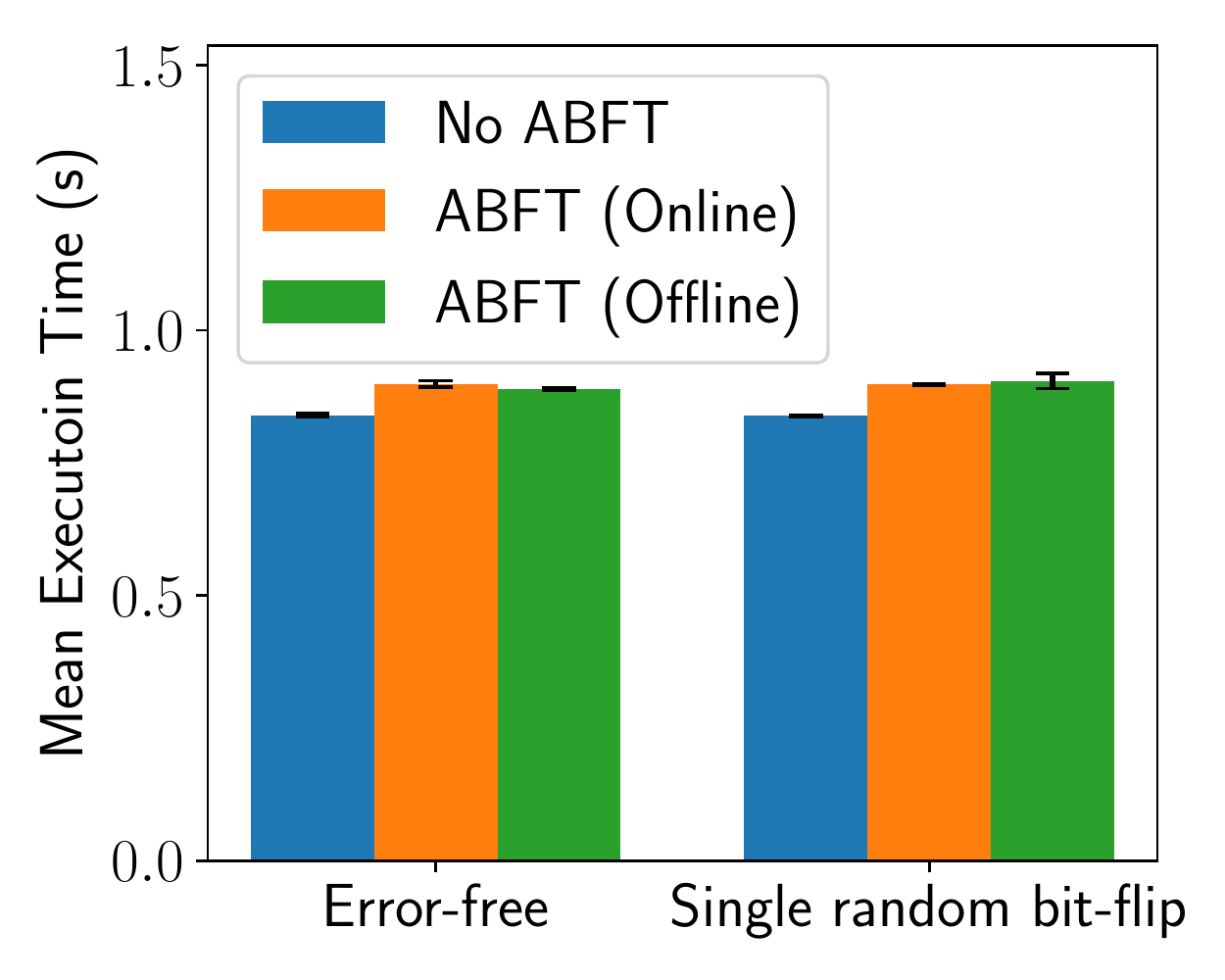}}
        \caption{Tile size \textbf{512$\times$512$\times$8}}
      \end{subfigure}
  \end{center}
  \caption{Mean execution time and standard deviation of the three ABFT methods for the HotSpot3D stencil with tiles of size (a)~\textbf{64$\times$64$\times$8} and (b)~\textbf{512$\times$512$\times$8}. The results show the error-free and single error cases during the execution, respectively.}
  \label{fig.overhead.float}
  \vspace{-0.1cm}
\end{figure}


\begin{figure}[!htb]
    \begin{center}
      \begin{subfigure}{\linewidth}
      	\centering
        \includegraphics[scale=0.52]{{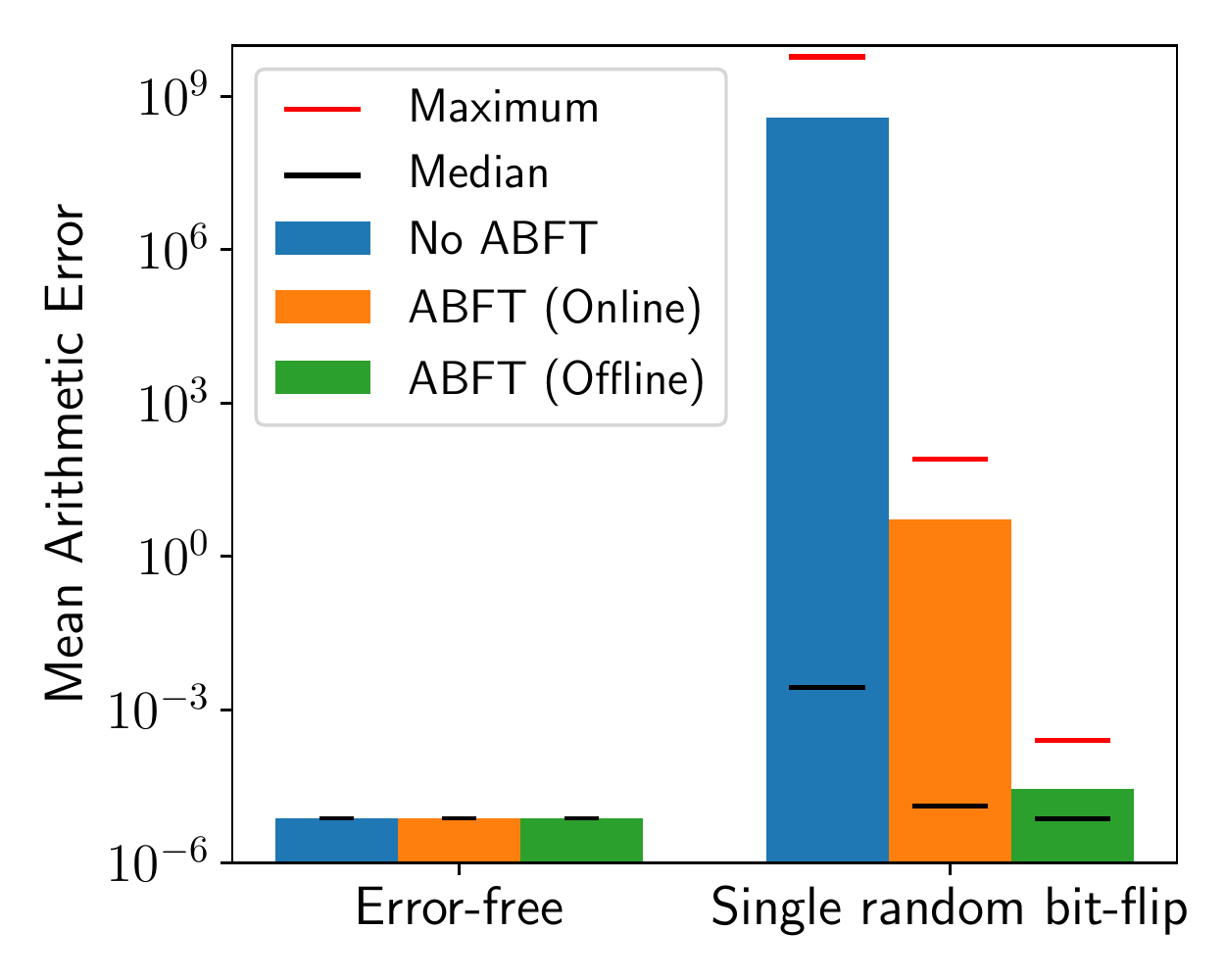}}
        \caption{Tile size \textbf{64$\times$64$\times$8}}
      \end{subfigure}
      \begin{subfigure}{\linewidth}
      	\centering
        \includegraphics[scale=0.52]{{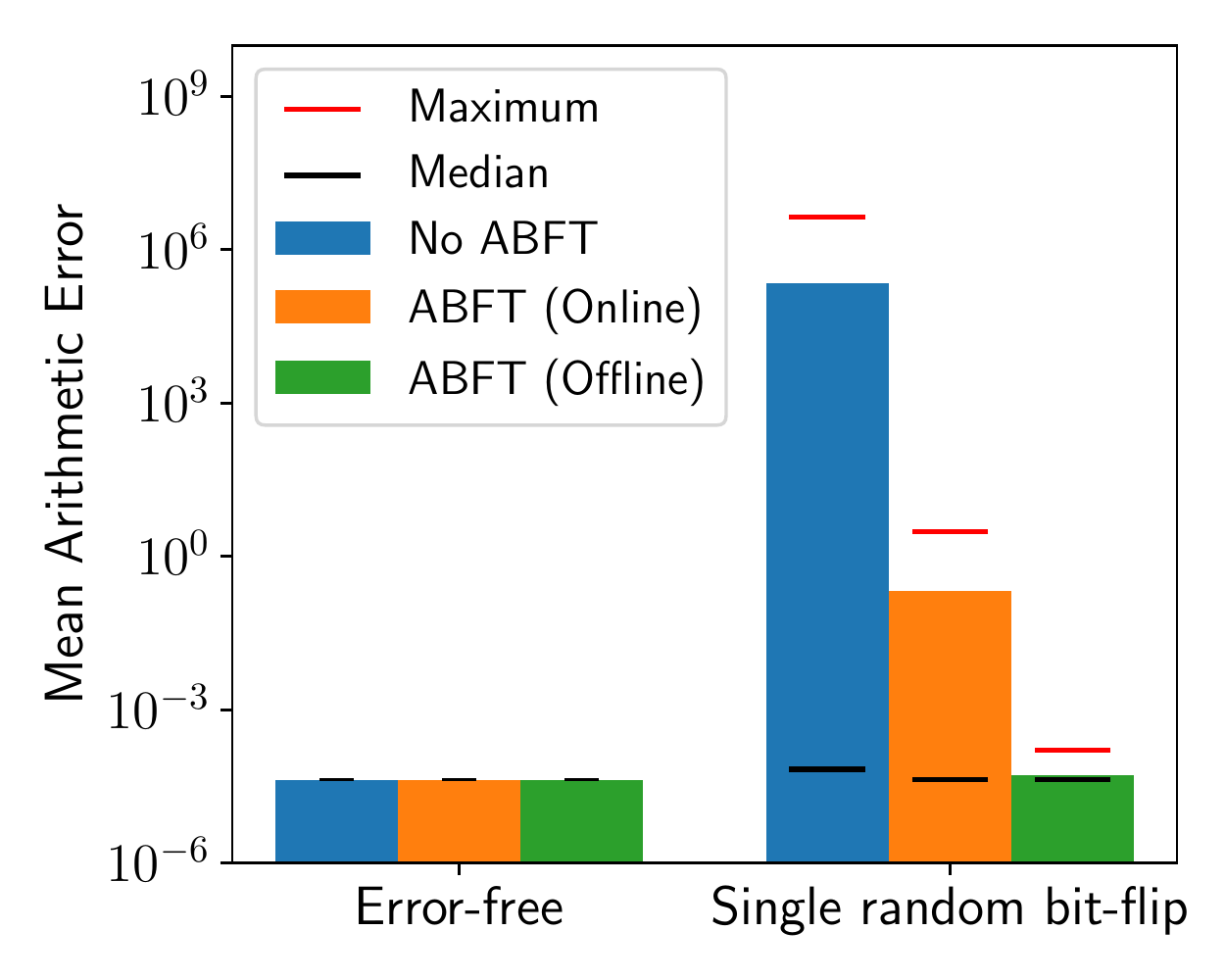}}
        \caption{Tile size \textbf{512$\times$512$\times$8}}
      \end{subfigure}
  \end{center}
  \caption{Mean, median, and maximum arithmetic error for the proposed ABFT methods for the HotSpot3D stencil with tiles of size (a)~\textbf{64$\times$64$\times$8} and (b)~\textbf{512$\times$512$\times$8}, in the error-free and error-prone (with a single random bit-flip) execution scenarios.
The arithmetic error, in both cases, is determined by comparing the results to the reference value (error-free, single threaded execution) using an $\textit{l}^{2}$-norm (see Equation~\eqref{eq.l2norm}).}
  \label{fig.error.float}
  \vspace{-0.1cm}
\end{figure}

Figure~\ref{fig.overhead.float} shows that in an error-free scenario, the proposed ABFT implementations are slightly more expensive than the original, non-protected implementation, with less than $8\%$ overhead for tile sizes $512\times512\times8$.
Note that the \emph{Offline ABFT} with a detection period of $16$ iterations and the \emph{Online ABFT} implementations have similar execution times, on average.
However, with a single randomly injected bit-flip during execution, the \emph{Offline ABFT} implementation becomes significantly slower compared to the \emph{Online ABFT} implementation.

Figure~\ref{fig.error.float} shows the corresponding arithmetic error on a logarithmic scale. 
In the error-free scenario, both methods lead to a total final error smaller than $10^{-5}$ compared to that of a single-threaded execution.
With a single bit-flip, the unprotected \emph{No-ABFT} original HotSpot3D implementation reaches very high mean and median error values, meaning that the final results are often corrupted beyond any acceptable threshold.
In comparison, our protected \emph{Online ABFT} and \emph{Offline ABFT} implementations in HotSpot3D keep the median error below $10^{-4}$.

While the \emph{Online ABFT} implementation yields acceptable mean and median error values, the \emph{Offline ABFT} scheme completely cancels out the error in most cases.
This is due to the chosen checkpointing and recovery strategy for the \emph{Offline ABFT}, which recomputes all values from the last checkpoint, while the \emph{Online ABFT} implementation corrects the values on-the-fly using the checksum vectors, which typically lead to a small approximation error.

\subsection{Evaluation of the Correction Accuracy}
\label{sec.accuracy}

In this section, experiments are conducted by fixing the position of the bit-flip within the 32-bit floating-point number representation. 
For each bit position, a bit-flip is randomly injected into the computational domain, during a random stencil iteration. 
In total, $1,000$ bit-flips have been injected per experiment, for every possible bit position.

\begin{figure}[!]
    \begin{center}
      \begin{subfigure}{\linewidth}
          \centering
        \includegraphics[scale=0.48]{{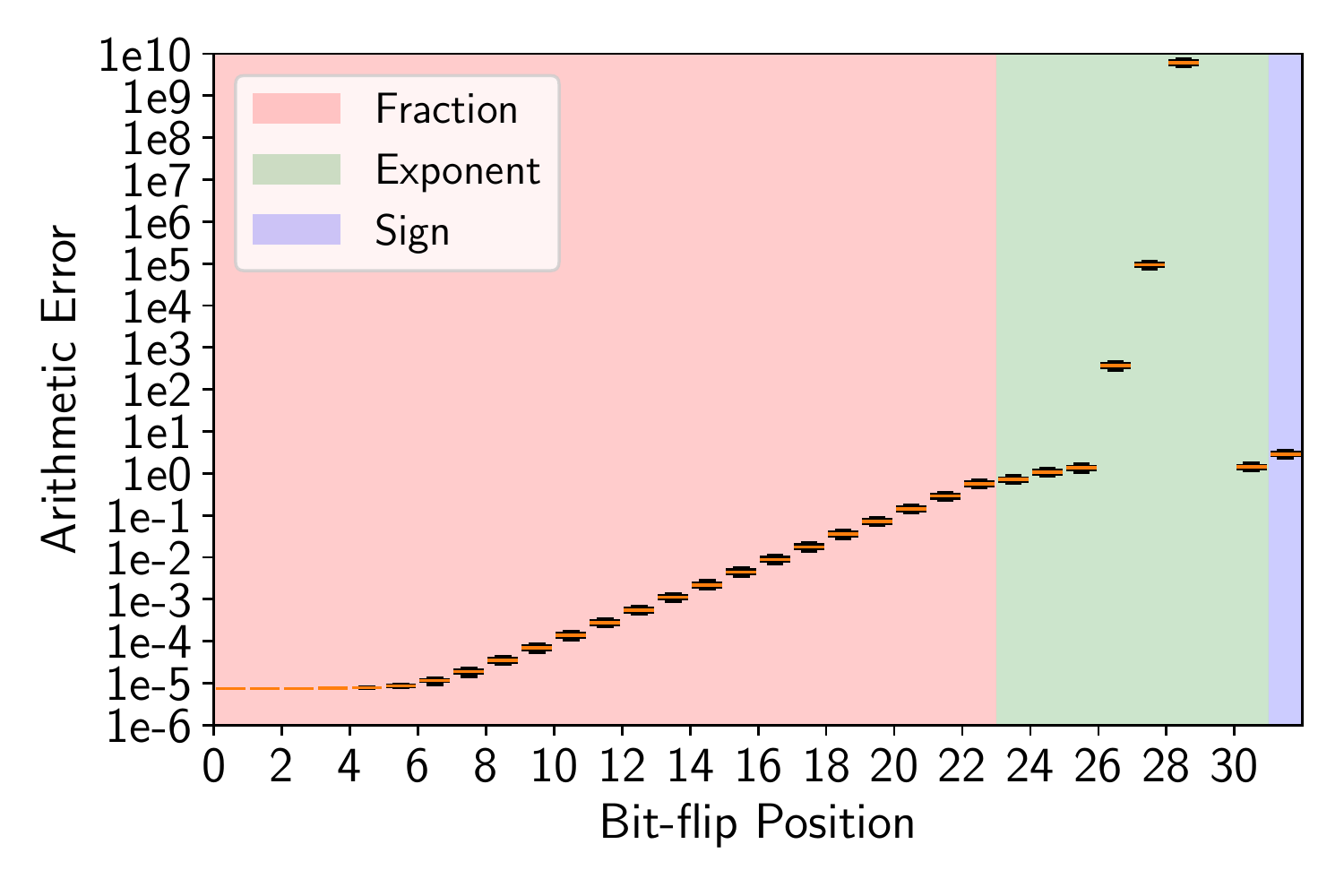}}\vspace{-0.35cm}
        \caption{No ABFT}\label{fig.biterror.noabft}
      \end{subfigure}
      \begin{subfigure}{\linewidth}
          \centering
        \includegraphics[scale=0.48]{{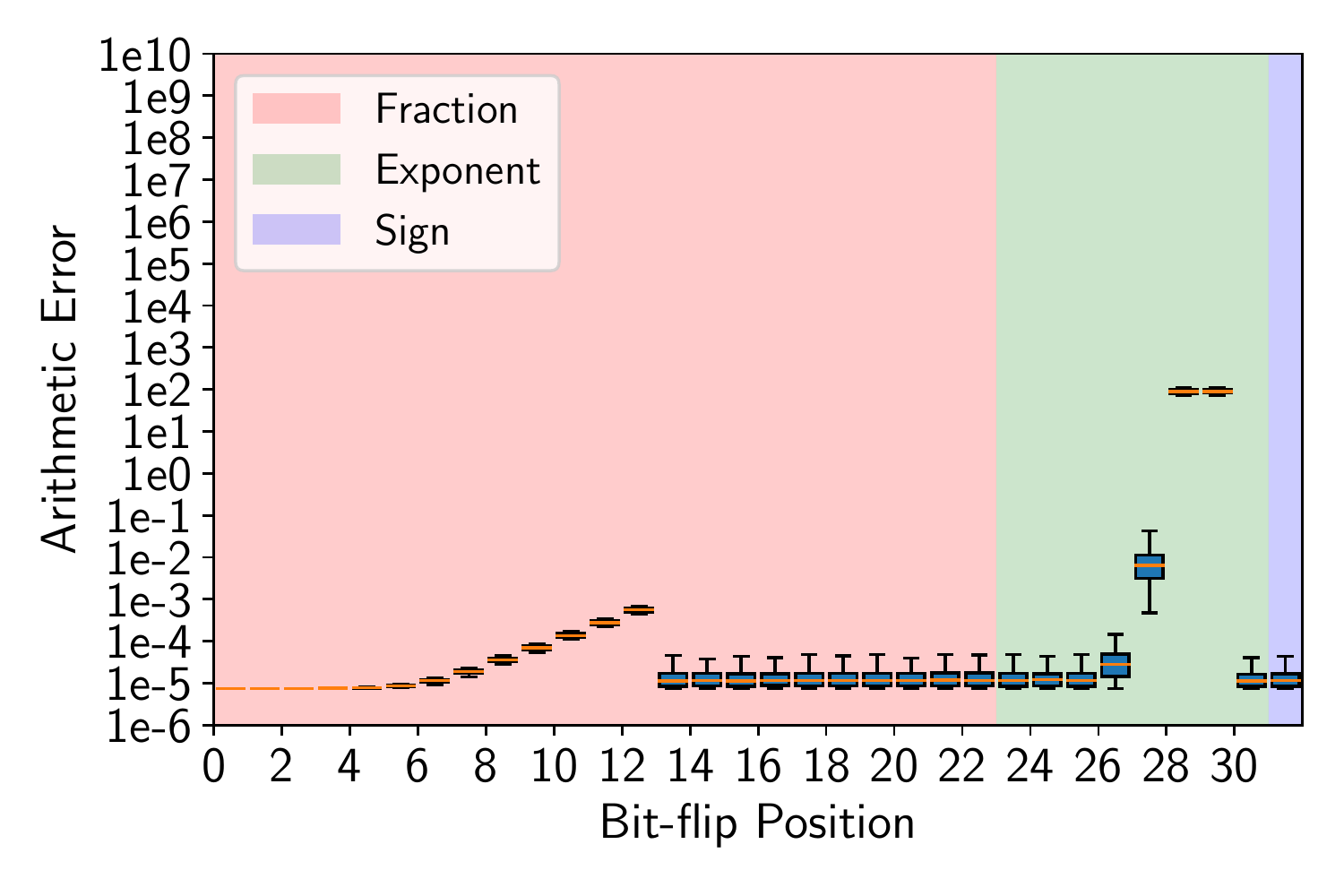}}\vspace{-0.35cm}
        \caption{Online ABFT}\label{fig.biterror.abft}
      \end{subfigure}
      \begin{subfigure}{\linewidth}
          \centering
        \includegraphics[scale=0.48]{{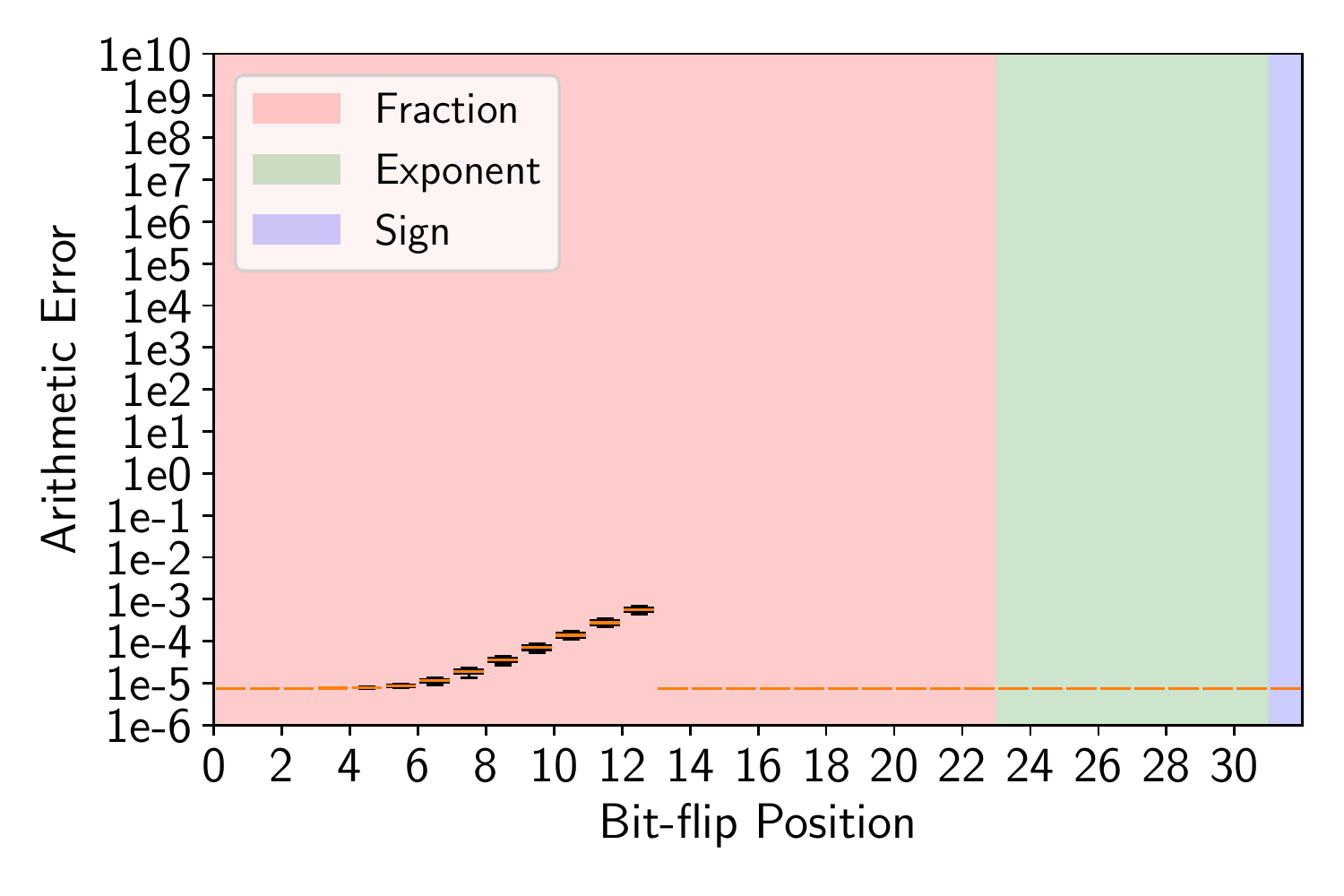}}\vspace{-0.35cm}
        \caption{Offline ABFT}\label{fig.biterror.abft.offline}
      \end{subfigure}
    \end{center}
  \caption{Impact of the bit-flip position on the error at the end of the execution (128 stencil iterations).
  The error in both cases is obtained by comparing the results to the reference value (error-free, single threaded execution) using an $\textit{l}^{2}$-norm (see Equation~\eqref{eq.l2norm}).
  Boxes show the interquartile range (Q3-Q1), where 50\% of the data points lie, and whiskers extend to 75\%. 
  The median error is indicated by an orange line.}
  \label{fig.biterror.float}
  \vspace{-0.1cm}
\end{figure}

Figure~\ref{fig.biterror.noabft} shows the results for the \emph{No-ABFT} implementation, which represents the baseline error when no detection and no correction mechanisms are in place. 
\ac{We observe that bit-flips in the exponent or the sign lead to very high errors in the final results.}

Figure~\ref{fig.biterror.abft} presents the results for the \emph{Online ABFT} implementation. 
While most bit-flips on bit positions between $13$ and $31$ can be detected and corrected, there is always a small residual error, albeit small enough to not cause significant degradation of the final result, in most cases.
However, if the bit-flips occur in the last bits of the exponent, the error causes an overflow in the checksums, which prevents the ABFT method to provide an accurate correction.

Figure~\ref{fig.biterror.abft.offline} shows the results for the \emph{Offline ABFT} implementation. 
Given that the \emph{Offline ABFT} scheme relies on checkpointing and recomputation in case of error, it is able to fully `erase' any error caused by the bit-flip, provided it can be detected.

For both \emph{Online ABFT} and the \emph{Offline ABFT} implementations, any bit-flip in a bit-position between $0$ and $12$, does not cause an error that is large enough to be detected.

\subsection{Impact of the Detection Period on \emph{Offline ABFT}}
\label{sec.interval}

In this section, we investigate the impact of the error detection period on the performance of the \emph{Offline ABFT} implementation. 
Figure~\ref{fig.interval} presents the mean execution time for tile sizes $64\times64\times8$~(Figure~\ref{fig.interval.64}) and $512\times512\times8$~(Figure~\ref{fig.interval.512}) in the error-free scenario and with a single random bit-flip injected during the execution.

\begin{figure}[!htb]
    \begin{center}
      \begin{subfigure}{0.493\linewidth}
        \includegraphics[scale=0.37]{{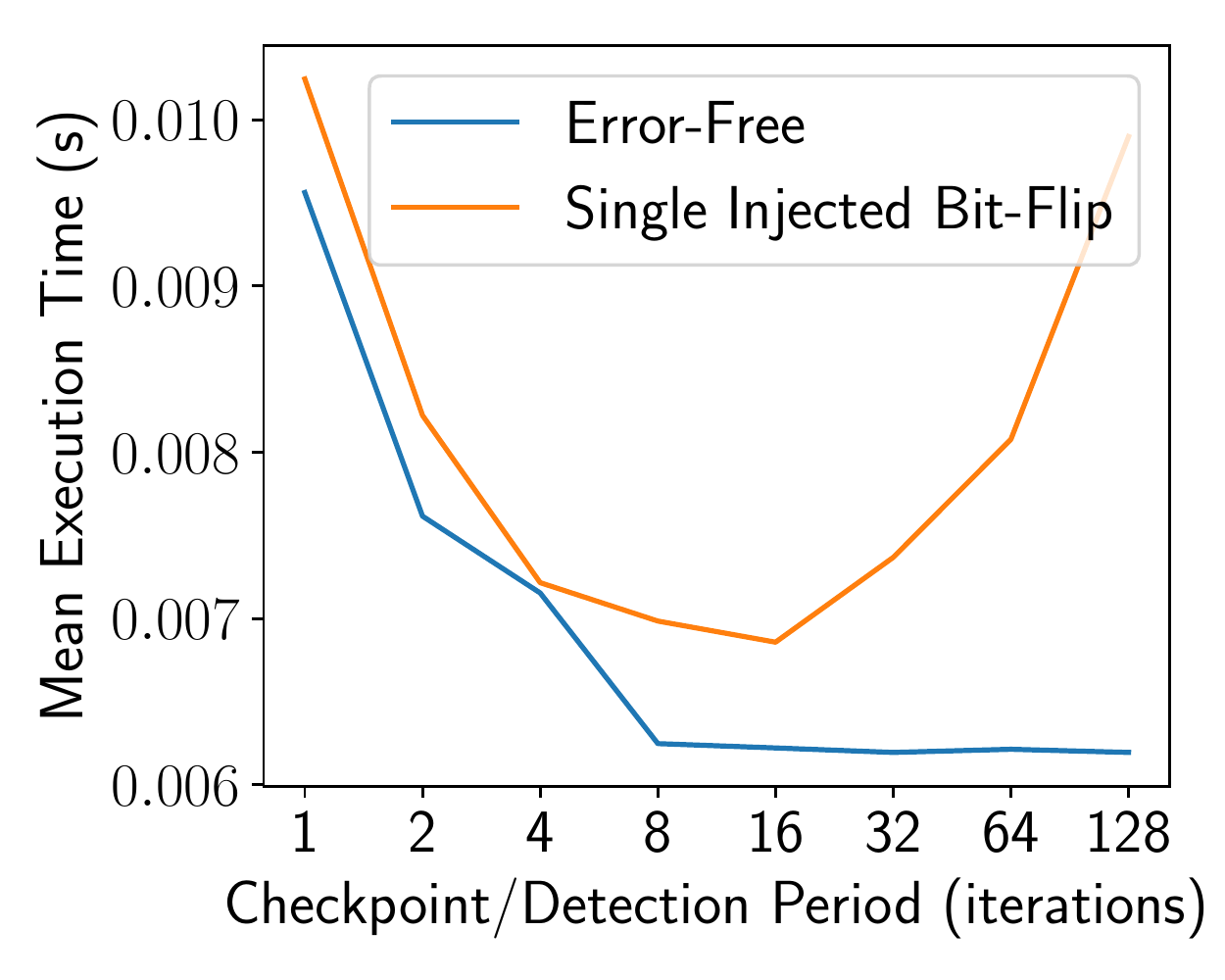}}
        \caption{Tile size \textbf{64$\times$64$\times$8}}\label{fig.interval.64}
      \end{subfigure}
      \begin{subfigure}{0.493\linewidth}
        \includegraphics[scale=0.37]{{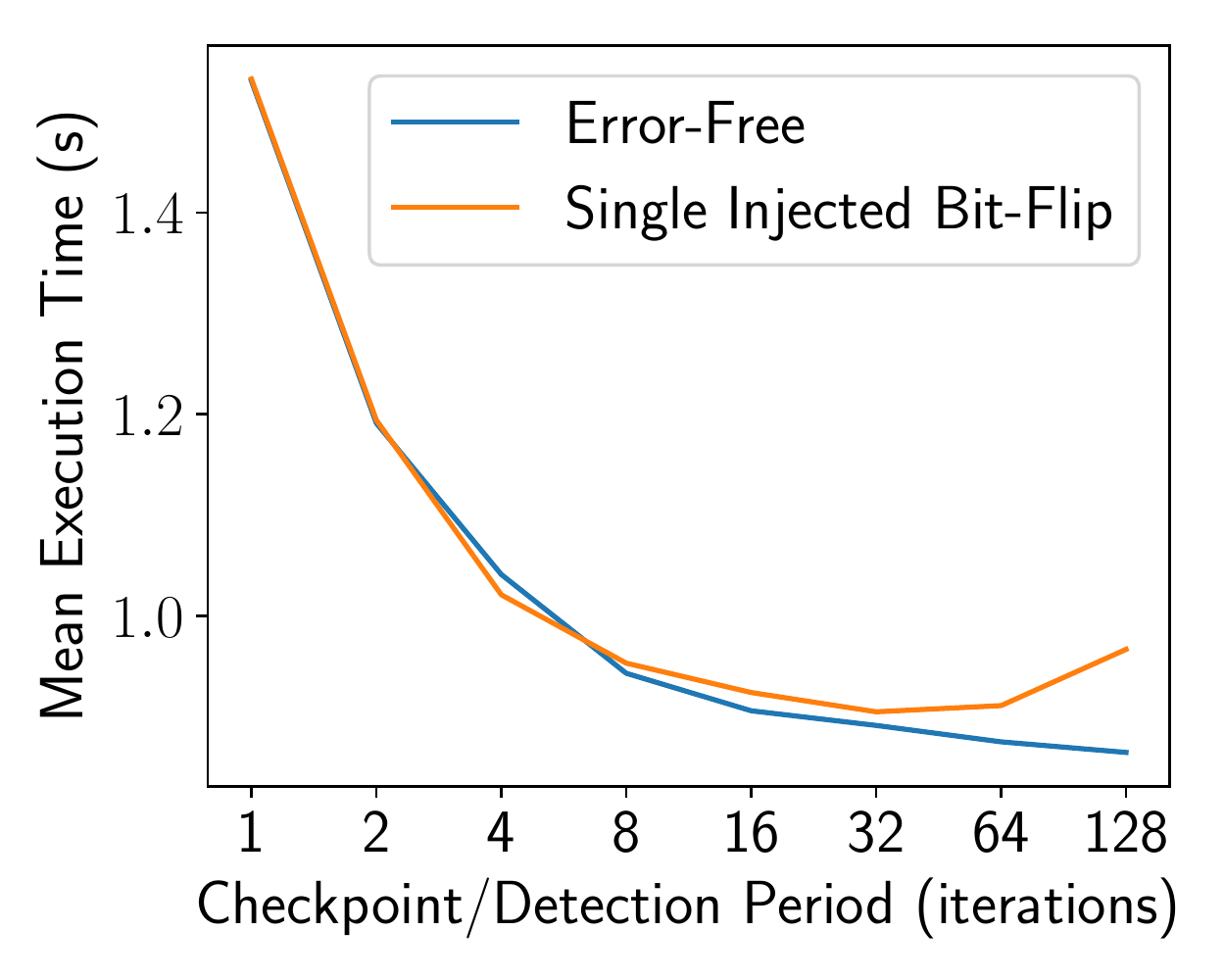}}
        \caption{Tile size \textbf{512$\times$512$\times$8}}\label{fig.interval.512}
      \end{subfigure}
    \end{center}
  \caption{Mean execution time and standard deviation for the proposed \emph{Offline ABFT} method with different detection intervals. 
  Results are presented for (a)~an error-free execution and (b)~with a single injected bit-flip during the execution.}
  \label{fig.interval}
\end{figure}

Experiments were conducted with error detection periods ranging from $0$ to $128$ iterations. 
The first, most important slowdown factor is the checkpointing that must be performed every $\Delta$ iterations. 
In these experiments, we perform a lightweight memory copy of the current state of the grid and of the checksums every $\Delta$ iterations, which has a constant cost. Therefore, the  checkpointing cost can be amortized over a longer period.
Other factors, such as the loop initialization cost or the vectorization capabilities of the system, can also negatively affect performance if the error detection period is too short, but have a negligible effect beyond a period of $16$ stencil iterations.
Overall, a detection period of $8$ or $16$ iterations yields the best ABFT performance in \emph{both} scenarios for HotSpot3D.



\section{Conclusion and Future Work}
\label{sec.conclusion}

In this work, we proposed a novel ABFT method \ac{to detect and correct silent data corruptions in} arbitrary stencil computations on 2D and 3D grids. 
The proposed ABFT scheme is made available both online, to be applied after every stencil iteration, as well as offline, to be applied at the end of the stencil computation or periodically \flo{during} given stencil iterations.
Experiments on a real 3D stencil application show that both the online and offline versions of the ABFT method incur less than $8\%$ overhead in error-free environments.
Fault-injection, detection, and correction experiments show that while the offline ABFT version generally provides near-perfect error detection and correction when coupled with checkpoint recovery, it may incur a slight overhead compared to the online ABFT version. 
Future work is aimed at adapting and evaluating the proposed ABFT method to different applications and hardware architectures, including GPUs and other accelerators, \flo{as well as combinations with other error correction approaches}. 

\section*{Acknowledgements}

This work has been supported by the Swiss Platform for Advanced Scientific Computing (PASC) project SPH-EXA: Optimizing Smooth Particle Hydrodynamics for Exascale Computing. The final publication is available at IEEE xplore.

\bibliographystyle{abbrv}
\bibliography{biblio}

\end{document}